\documentclass[conference]{IEEEtran}
\usepackage{amsthm}
\usepackage{amsmath}
\usepackage{algorithm}
\usepackage{algorithmic}
\usepackage{verbatim} 
\usepackage{amsthm}

\usepackage{color}
\usepackage{soul}
\usepackage{epsfig}
\usepackage{makeidx}
\usepackage{ifpdf}
\usepackage[table]{xcolor}
\usepackage{graphicx}
\usepackage{color}
\usepackage{psfrag}
\usepackage[margin = 0.7in]{geometry}
\usepackage{subfigure}
\label{before begin doc}

\IEEEoverridecommandlockouts

\begin{document}

\title{Optimal Power Control and Rate Adaptation for Ultra-Reliable M2M Control Applications}

\author{Bakhtiyar Farayev, Yalcin Sadi, \textit{Student Member}, \textit{IEEE}, and Sinem Coleri Ergen, \textit{Member}, \textit{IEEE}
\thanks{Bakhtiyar Farayev, Yalcin Sadi and Sinem Coleri Ergen are with the department of Electrical and Electronics Engineering, Koc University, Istanbul, Turkey, e-mail:~\texttt{bfarayev13, ysadi, sergen@ku.edu.tr}. Authors acknowledge the support of The Scientific and Technological Research Council of Turkey Grant \#113E233 and Turk Telekom Collaborative Research Award \#11315-10. Sinem Coleri Ergen also acknowledges support from Bilim Akademisi - The Science Academy, Turkey under the BAGEP program, and the Turkish Academy of Sciences (TUBA) within the Young Scientist Award Program (GEBIP).}
}

% make the title area
\maketitle

\newcommand{\diag}{\textrm{diag}}
\newtheorem{theorem}{\textbf{Theorem}}
\newtheorem{assumption}{\textbf{Assumption}}
\newtheorem{lemma}{\textbf{Lemma}}
\newtheorem{corollary}{\textbf{Corollary}}
\newtheorem{proposition}{\textbf{Proposition}}%[section]
\newtheorem{remark}{\textbf{Remark}}%[section]
\newtheorem{claim}{\textbf{Claim}}%[section]
\newtheorem{definition}{\textbf{Definition}}%[section]
\newtheorem{case}{\textbf{Case}}%[section]
\newtheorem{observation}{Observation}
\newenvironment{Proof}{ \emph{Proof:}}{$\Box$}
\newenvironment{example}{ \textbf{Example} \footnotesize}{\newline}
\def\E{\,\mathds{E}\,}
\def\Blue#1{\textcolor{blue}{#1}}
\def\Red#1{\textcolor{red}{#1}}
\def\Green#1{\textcolor{green}{#1}}

\IEEEpeerreviewmaketitle

\begin{abstract}
The main challenge of ultra-reliable machine-to-machine (M2M) control applications is to meet the stringent timing and reliability requirements of control systems, despite the adverse
properties of wireless communication for delay and packet errors, and limited battery resources of the sensor nodes. Since the
transmission delay and energy consumption of a sensor node are determined by the transmission power
and rate of that sensor node and the concurrently transmitting nodes, the transmission schedule should be optimized
jointly with the transmission power and rate of the sensor nodes. Previously, it has been shown that the optimization of power control and rate adaptation for each node
subset can be separately formulated, solved and then used in the scheduling algorithm in the optimal solution of the joint optimization of power control, rate adaptation and scheduling problem. However, the power control and rate adaptation problem has been only formulated and solved for continuous rate transmission model, in which Shannon's capacity
formulation for an Additive White Gaussian Noise (AWGN) wireless channel is used in the calculation
of the maximum achievable rate as a function of Signal-to-Interference-plus-Noise Ratio (SINR). In this paper, we formulate the power control and rate adaptation problem with the objective of minimizing
the time required for the concurrent transmission of a set of sensor nodes while satisfying their
transmission delay, reliability and energy consumption requirements based on the more realistic
discrete rate transmission model, in which only a finite set of transmit rates are supported. We propose a polynomial time
algorithm to solve this problem and prove the optimality of the proposed algorithm. We then combine it with the previously proposed scheduling algorithms and demonstrate its close to optimal performance via extensive simulations. 

\end{abstract}


\section{Introduction}

M2M is a new communication
paradigm that aims to enable the communication between the
devices with no or minimal human interaction. The wireless communication between sensor nodes and controllers in ultra-reliable M2M control applications reduces the cost of their installation, maintenance and upgrade, together with the complexity of the overall system, compared to their wired equivalent~\cite{Joao07}. 
%WNCSs therefore have been supported by several leading industrial organizations such as International Society of Automation (ISA)~\cite{isa_sp100} and Highway Addressable Remote Transducer (HART)~\cite{whart} with vast number of applications in industrial automation~\cite{industrial_automation}, building automation~\cite{Aswani12}, smart grid \cite{smart_grid} and vehicle control systems \cite{intravehicular}. 
The main challenge of ultra-reliable M2M control applications is the design of a robust scheduling algorithm that satisfies the packet generation period, transmission delay and reliability requirement of the sensors required to maintain a certain control system performance \cite{networkedcontrol-scheduling-12, networkedcontrol-scheduling-07} despite the adverse properties of wireless communication for non-zero packet error probability and non-zero delay at all times. 
%The scheduling algorithm should exploit the periodic transmission nature of the sensor nodes at pre known packet generation rates while providing robustness to packet losses and changes in network topology by distributing enough unused schedule periods over time allowing the inclusion of unexpected required additional packet transmissions as needed. Moreover, the schedule should satisfy the transmission energy requirement of the sensor nodes to achieve their desired lifetime, which are mostly battery operated or rely on energy harvesting techniques. 
Since the transmission delay and energy consumption of a sensor node are determined by the transmission power and rate of that sensor node and by the same of the other nodes that are scheduled for concurrently transmission, the schedule should be optimized jointly with the transmission power and rate of the sensor nodes.

Ultra-reliable communication (URC) is defined as a high-level communication service which is available almost $100\%$ of the time. \cite{ultra_01, ultra_02} discuss the relation between control information and actual data in URC based systems. \cite{ultra_03} proposes a URC based system based on the prediction of future reliable channels. None of these works however consider the periodic packet transmission of the sensor nodes in control applications.

Joint optimization of power control, rate adaptation and scheduling has been studied for both general purpose wireless networks \cite{fu2006optimal, ozel2012optimal, chen2008energy, zafer2009calculus} and wireless networked control systems (WNCSs) \cite{intravehicular, son_makale}. This optimization problem has been formulated and solved for general purpose networks with the goal of satisfying either a single deadline for all packets  \cite{fu2006optimal, ozel2012optimal} or individual deadlines for each packet \cite{chen2008energy, zafer2009calculus}. In addition to the delay constraints, the formulation of this joint optimization problem for WNCSs has included the uniform distribution of packet transmissions over time to satisfy the robustness requirement of the control systems, and the periodic data generation, reliability and energy consumption requirements of the sensor nodes \cite{intravehicular, son_makale}. It has been shown in \cite{son_makale} that the optimization of power control and rate adaptation for each node
subset can be separately formulated, solved and then used in the scheduling algorithm in the optimal solution of the joint optimization of power control, rate adaptation and scheduling problem. However, the power control and rate adaptation problem has been formulated for continuous rate transmission model \cite{intravehicular, son_makale}, in which Shannon's channel capacity
formulation for an AWGN wireless channel is used in the calculation
of the maximum achievable rate as a function of SINR. In practical communication systems,
however, only a finite set of discrete transmission rate values are supported, as suggested by the practical
realization of multiple data rates in \cite{goldsmith}.

The goal of this paper is to study the power control and rate adaptation problem with the objective of minimizing
the time required for the concurrent transmission of a set of sensor nodes while satisfying their
transmission delay, reliability and energy consumption requirements based on the more realistic
discrete rate transmission model. We propose a polynomial time
algorithm for the power control and rate adaptation problem and provide the proof
of the optimality of the proposed algorithm. We then demonstrate the performance of the proposed algorithm within the previously proposed scheduling framework for ultra-reliable M2M control via extensive simulations.

The rest of this paper is organized as follows. Section \ref{system_model} describes the system model and assumptions used throughout this paper. Section \ref{power_rate} formulates the power control and rate adaptation problem, proposes a polynomial time algorithm and proves the optimality of the proposed algorithm. Performance results of the proposed algorithm are presented in Section \ref{simulationresults}. The final review of our work and the ideas for future research are given in Section \ref{sec:conclusion}.

\section{System Model and Assumptions} \label{system_model}
\begin{enumerate}
\item The network contains sensor nodes, controllers and actuators. %The sensor nodes fall into one of two categories: time-triggered sensors and event-driven sensors. Time-triggered sensor nodes sample and send their data to one of the controllers periodically. The packet generation period and transmission delay required to achieve a certain control performance are denoted by $T_l$ and $d_l$, respectively, for a time-triggered sensor node $l$. On the other hand, event-driven sensor nodes sample and send the data to the corresponding controller upon a triggering event in the unallocated parts of the schedule. 
The controllers can receive packets only from one sensor node at a time. When a controller receives the sensor data, it performs the new control computation and sends the output to the actuator. We assume that the actuators always receive the controller packets successfully.
%\item We use the term sensor node as either \textit{node} or \textit{sensor} for simplicity. Also, the term \textit{link} will be used for the connection between the node and its associated controller.
\item We assume that sensor nodes transmit their data packets to their corresponding controllers in one hop.
The extension for the multi-hop packet transmission of the packets from the sensor nodes to the controllers is out of scope of this paper and subject to future work.
\item Central controller is selected from the existing controllers and responsible for network synchronization, resource allocation and scheduling of the sensor nodes. Central controller is assumed to have complete network topology, however the method for topology learning is out of scope of this paper. 
\item TDMA is used as a Medium Access Control (MAC) protocol due to its superior delay and energy performance for the networks with predetermined topology and data generation pattern \cite{pedamacs}. The time is divided into subframes, which are further divided into beacon and time slots. The central controller sends the beacon at the beginning of each subframe and includes the scheduling and resource allocation information of each scheduled node within the subframe. If there exist sensor nodes outside the transmission range of the central controller then the central controller uses the corresponding controllers to relay the beacon to these nodes. %\item We assume that the packet generation periods are multiples of each other, e.g. $T = \{1, 2, 4, 8\}$ ms. Periodic generation of sensor data at various frequencies exist in many WNCS applications such as industrial automation~\cite{industrial_automation}, smart grid \cite{smart_grid} and vehicle control \cite{intravehicular}. 
%\item We assume that the length of the time slot allocated to a sensor node is fixed over all the subframes it has been allocated to and the time difference between the consecutive time slots assigned to each node is fixed and equal to its packet generation period. This assumption guarantees the fixed determinism of control systems, which is frequently preferred to bounded determinism \cite{networkedcontrol-scheduling-05}. 
\item We only consider the energy consumption for the data packet transmission of each sensor node since the energy consumption in active mode is much larger than that in sleep and transient modes, and the energy consumption in the reception of the beacon packets is not subject to optimization. We define the maximum allowed per packet energy consumption to achieve a certain lifetime by $e_l$ for sensor node $l$.
\item The transmission power can take any value below $P_{max}$. This continuous power assumption is frequently used in previous scheduling algorithm designs since practical radios support a large number of discrete power levels, resulting in high approximation accuracy with lower complexity.
\item We assume that the channel is perfectly known and constant over the scheduling frame, i.e. channel gain $g_{lk}$ from the sensor node $l$ to the controller corresponding to sensor node $k$ does not change. Considering the channel variation of the links within the scheduling frame so fast fading wireless channels is out of scope of this paper and subject to future work.
\item We use the \textit{discrete rate} transmission model in which a finite number of transmission rates $r = (r^1, r^2, ..., r^Q)$ corresponding to a finite number of SINR levels $\gamma = (\gamma^1, \gamma^2, ..., \gamma^Q)$ are determined such that node $l$ can transmit at rate $r^q$ if the SINR achieved at the corresponding link, i.e. $\gamma_l = \frac{p_lg_{lk}}{N_0 + \sum_{l\neq k} p_lg_{lk}}$, is greater than or equal to $\gamma^q$, where $p_l$ is the transmit power of node $l$, $N_0$ is the background noise power and $Q$ is the number of discrete transmission rate levels. Each rate level $r^i$, $i \in [1,Q]$ is determined based on SINR-rate mapping function $f(\gamma)$, i.e. $r^i = f(\gamma^i)$, $i \in [1,Q]$, which satisfies the constraint given by $\frac{\partial^2 f(\gamma)}{\partial\gamma^2} \leq 0$. This constraint is satisfied for Shannon's capacity formulation which is commonly used in AWGN channels.
\item The scheduling algorithms proposed for WNCS in \cite{son_makale} has been adopted to be jointly used with the proposed power control and rate adaptation algorithm. The goal of the scheduling algorithm is to provide maximum robustness accommodating packet losses, topology and sensor requirement changes. 
The subframe and frame length are defined as the minimum and maximum packet generation period of all sensor nodes, and denoted by $S$ and $F$, respectively. The \emph{total active length of a subframe} $m$, denoted by $a_m$, is defined as the sum of the length of the time slots allocated to subframe $m$. The robustness of the schedule requires distribution of data transmissions as uniformly as possible over the scheduling frame. The objective is therefore quantified as minimizing the maximum total active length of all the subframes in a frame. 

To illustrate the scheduling, let us assume the example network consist of $4$ sensor nodes. Sensor node $1$ has packet generation period of $1$ ms and packet transmission time of $t_1=0.15$ ms. Sensor nodes $2$, $3$ and $4$ have packet generation period of $2$ ms and their packet transmissions require time slot allocations of lengths $t_2=0.20$ ms, $t_3=0.25$ ms and $t_4=0.30$ ms, respectively, when no concurrent transmissions are allowed. Let us assume that the concurrent transmission of nodes $2$ and $3$ takes $0.30$ ms, i.e. $t_{2,3} = 0.30$ ms. 
Here, the subframe length $S=1$ ms and the frame length $F=2$ ms, containing $2$ subframes.  
%Fig. \ref{fig:objective} (a) shows the robust schedule that minimizes the maximum total active length of the subframes. The maximum total active length is equal to $max(a_1,a_2) = max(t_1+t_2+t_3, t_1+t_4) = 0.60$ ms. 
Fig. \ref{fig:objective} shows the robust schedule that minimizes the maximum total active length of the subframes in the presence of the concurrent transmissions. The maximum total active length is equal to $max(a_1,a_2) = max(t_1+t_{2,3}, t_1+t_4) = 0.45$ ms. 

\begin{figure}[H]
  \begin{center} 
    \includegraphics*[scale=0.50,viewport=30 30 490 120]{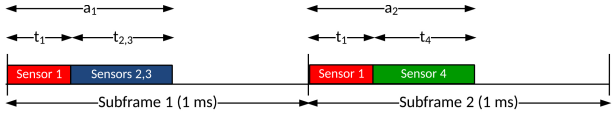}
  \end{center}  
  \caption{\small Illustration of the robust scheduling algorithm for the network of 4 nodes.}  
  \label{fig:objective}
\end{figure}

The scheduling algorithms proposed in \cite{son_makale} are based on first assigning the nodes to the subframes through the node assignment algorithm and then determining the subsets of nodes that concurrently transmit through the concurrency allocation algorithm. The first algorithm adopts Sorted Node Assignment (SNA) and Minimum Length Allocation (MLA) hence is denoted by SNA-MLA algorithm, whereas the second algorithm adopts SNA and Maximum Utility Allocation (MUA) therefore is called SNA-MUA. SNA is based on assigning each node to the subframe with the minimum total active length giving priority to the nodes with higher transmission times and considering the nodes connected to different controllers separately. MLA algorithm enumerates all feasible subsets of the nodes with the same packet generation period assigned to the same subframe, and then selects some of these feasible subsets such that every node is covered at least once and the total time allocated is minimized. MUA algorithm, on the other hand, iteratively constructs the concurrently transmitting node subsets from the nodes with the same packet generation period assigned to the same subframe by including the node that maximizes the utility, which is defined as the decrease in the transmission time of a set of sensor nodes when they transmit concurrently, in the subset. The SNA algorithm uses power control and rate adaptation in determining the transmission time of each sensor node while satisfying its delay and energy requirement whereas the MLA and MUA algorithms use power control and rate adaptation algorithm in deciding whether the concurrent transmission of a subset of sensor nodes while satisfying their transmission delay and energy requirements is feasible and their transmission time if feasible. 

\end{enumerate}

\section{Power Control and Rate Adaptation Problem} \label{power_rate}
\subsection{Problem Formulation}

The power control and rate adaptation problem aims to minimize the time required for the concurrent transmission of the sensor nodes, while satisfying their transmission delay and energy consumption requirements. As described in detail in Section \ref{system_model}, this problem needs to be solved to provide the transmission time of a single node and subset of nodes in the node assignment and concurrency allocation parts, respectively, of the SNA-MLA and SNA-MUA algorithms. The formulation of the power control and rate adaptation problem for the sensor nodes in the set $S$ is given as

\begin{subequations} \small

\textbf{minimize}
\begin{equation} \label{eq:t}
t
\end{equation}

\textbf{subject to}
\begin{equation} \label{eq:tileq}
t_l \leq t, \quad l\in S
\end{equation}

\begin{equation} \label{eq:di}
t_l \leq d_l, \quad l\in S 
\end{equation}

\begin{equation} \label{eq:ei}
t_lp_l \leq e_l, \quad l\in S
\end{equation}

\begin{equation} \label{eq:pmax}
p_l \leq p_{max}, \quad l\in S
\end{equation}

\begin{equation} \label{eq:ti}
t_l = \frac{R_l}{x_l}, \quad l\in S
\end{equation}

\begin{equation} \label{eq:pow}
p_l g_{ll} - k_{lq}\gamma^q (N_0 + \sum_{u\neq l}p_ug_{ul}) \geq 0, \quad  l\in S, q\in[1,Q] 
\end{equation}

\begin{equation} \label{eq:k}
\sum_{q=1}^Q k_{lq} = 1, \quad l \in S 
\end{equation}

\begin{equation}  \label{eq:r}
x_l = \sum_{q=1}^{Q}k_{lq} r^q, \quad l \in S
\end{equation}

\textbf{variables}
\begin{center}
%$t\geq 0$, 
$t>0$, $p_l \geq 0$, $k_{lq} \in \{0,1\}$, $l\in S$, $q\in[1,Q]$
\end{center}

\end{subequations}

where $t$ is the time duration required for the concurrent transmission of the nodes in the set $S$; $t_l$ and $x_l$ represent the transmission time and the transmission rate of node $l$, respectively; $R_l$ is the length of the packet of sensor node $l$; $k_{lq}$ is the indicator of the assignment of sensor node $l$ to the $q$-th rate level; $d_l$ is the delay requirement of sensor node $l$ to achieve a certain control performance. Eqs. (\ref{eq:t}) and (\ref{eq:tileq}) together represent the objective of minimizing the time slot length allocated for the concurrent transmission of the nodes, i.e. minimizing maximum transmission time of the nodes. Eqs. (\ref{eq:di}) and (\ref{eq:ei}) represent the delay and energy constraints of the nodes, respectively. Eq. (\ref{eq:pmax}) states that the transmit power of the $l$-th node cannot be greater than the maximum allowed power. Finally, Eqs. (\ref{eq:ti}) - (\ref{eq:r}) represent the computation of the transmission time of node $l$ by assigning one of the discrete transmission rates in the set $r$ that satisfy the corresponding SINR requirement.

Next, we propose a polynomial time algorithm for the solution of this optimization problem and give a detailed proof of the optimality of the proposed algorithm.

%In order to solve this optimization problem in polynomial-time complexity, the elimination method will be used in the algorithm, given that $S$ number of nodes are concurrently transmitting. The transmission time vector will be denoted as $\textbf{t}$, where $\textbf{t} = (t_1, t_2, ..., t_S)$ and $t_{max} = max_{i\in[1,|S|]}t_i$. For simplicity, we will use $t$ instead of $t_{max}$ since its equivalent to the transmit time of concurrently transmitting $S$ nodes. 
%

\subsection{Proposed Polynomial Algorithm}

We propose Longest Transmission Time First (LTTF) algorithm to determine the optimal power and rate allocation for the concurrent transmission of the nodes in the set $S$ in polynomial time. The algorithm is based on the initialization with the vector of the lowest possible rates satisfying the delay constraints of the nodes and then checking the feasibility of the lowest possible number of rate allocation vectors for satisfying energy, maximum power and SINR constraints. 

Let us enumerate the links in the set $S$ from $1$ to $|S|$. Let $\textbf{r}$ be a vector of the transmission rates of the links in the set $S$.  Determining the feasibility of a rate vector $\textbf{r}$ requires checking the existence of a power vector corresponding to $\textbf{r}$, delay, energy and maximum power constraints. Since the rates of the links are predetermined by the corresponding elements of the vector $\textbf{r}$, the corresponding SINR threshold levels are predetermined. Then the existence of a transmit power vector that satisfies the constraint
\begin{equation}
\frac{p_ig_{ii}}{N_0 + \sum_{j\neq i} p_jg_{ji}} > \gamma_i
\end{equation} 
for every node $i$, where $\gamma_i$ is the SINR level corresponding to $r_i$, can be found by testing Perron-Frobenius conditions \cite{PF}. Perron-Frobenius conditions determine whether there exist a power vector satisfying these SINR constraints, and if so, the component-wise minimum power vector $\textbf{p}^{min}$. $\textbf{p}^{min}$ is then tested for delay, energy and maximum power constraints given by Eqs. (\ref{eq:di}), (\ref{eq:ei}) and (\ref{eq:pmax}), respectively.

Let $t$ denote the time duration required for the concurrent transmission of the nodes in the set $S$. The optimal values of $\textbf{r}$ and $t$ are denoted by $\textbf{r}^{opt}$  and $t^{opt}$, respectively. The $i$-th element of the vector $\textbf{r}^{opt}$ corresponding to the optimal rate of node $i$, denoted by $r^{opt}_i$, is initialized to the minimum rate level that satisfies the delay constraint of node $i$, denoted by $r_i^{th}$, given by Eq. (\ref{eq:di}) for all $i \in [1, |S|]$, whereas $t^{opt}$ is initialized to infinity (Lines $1-2$). If this initial $\textbf{r}^{opt}$ is not feasible, then the algorithm returns $ t^{opt} = \infty$, meaning that there does not exist a feasible solution (Lines $3$ and $15$). Otherwise, the algorithm calculates the time duration of the time slot corresponding to the rate vector $\textbf{r}^{opt}$ by determining the maximum transmission time of the nodes in the set $S$ (Lines $4-5$). Since the node with the maximum transmission time denoted by $j$ determines the time slot length, the length of the time slot can only decrease by increasing its transmission rate. Therefore, the algorithm continues by increasing the rate of the $j$-th node by one level if it is less than the maximum possible rate $r^Q$ (Lines $6-7$), and then checking the feasibility of the resulting $\textbf{r}^{opt}$ vector (Line $3$). The algorithm stops if either the $j$-th node has rate $r^Q$, meaning that it is not possible to decrease the time slot length any further (Lines $8-9$), or the resulting $\textbf{r}^{opt}$ vector is not feasible. 

 The complexity of the LTTF algorithm is $O(Q|S|^4$) since the algorithm checks at most $Q|S|$ transmission rate vectors for feasibility and the complexity of the feasibility check is dominated by the complexity of testing Perron-Frobenius conditions of complexity $O(|S|^3)$ \cite{PF}.
 
\begin{algorithm}
\caption{Longest Transmission Time First (LTTF) algorithm}
\label{alg1}
\begin{algorithmic}[1]

\REQUIRE Node set $S$ considered for concurrent transmission
\ENSURE Optimal rate values for the nodes in the set $S$ if their concurrent transmission is feasible 
\STATE $ t^{opt} = \infty;$
\STATE $\textbf{r}^{opt} = (r_1^{th}, r_2^{th}, ... , r_{|S|}^{th})$, \\where $r_i^{th} = min_{q \in [1,Q]}\{r^q| \frac{R_i}{r^q} \leq d_i\}$;
\WHILE {$\textbf{r}^{opt}$ is feasible}
\STATE $j = argmax_{i \in [1,|S|]}\frac{R_i}{r_i^{opt}}$;
\STATE $t^{opt} = \frac{R_j}{r_j^{opt}}$;
\IF {$r_j^{opt} < r^Q$}
\STATE increase $r_j^{opt}$ by one level;
\ELSE
\STATE break;
\ENDIF
\ENDWHILE
\IF {$t^{opt} < \infty$ and $\textbf{r}^{opt}$ is not feasible}
\STATE decrease $r_j^{opt}$ by one level;
\ENDIF
\STATE \textbf{return} $r^{opt}$, $t^{opt}$
\end{algorithmic}
\end{algorithm}

\subsection{Optimality of Proposed Algorithm}

We now prove the optimality of the proposed algorithm following the statement of the following three lemmas. 

Let $\textbf{r}=(r_1, r_2, ..., r_{|S|})$ be a vector of the transmission rates of the links in the set $S$ and $t$ be the length of the time slot corresponding to the rate vector $\textbf{r}$, such that $t = max_{i\in[1,|S|]}\frac{R_i}{r_i}$, if $\textbf{r}$ is feasible. Denote the rate index vector $\textbf{r}^d = (r_1^d, r_2^d, ..., r_{|S|}^d)$, such that for at least one element $j$, $r_j^d < r_j$, while $r_{i}^d \leq r_{i}$ for  all $i\in[1,|S|]$, by the descendant of $\textbf{r}$.

\label{lemma3}
\begin{lemma}
Let both $\textbf{r}$ and $\textbf{r}^d$ satisfy the delay constraint given in Eq. (\ref{eq:di}). If $\textbf{r}^d$ is infeasible, then $\textbf{r}$ is infeasible.
\end{lemma}
\begin{proof}
Let us denote the component-wise minimum power vectors corresponding to the rate vectors $\textbf{r}$ and $\textbf{r}^d$ by $\textbf{p} = (p_1, p_2, ..., p_{|S|})$ and $\textbf{p}^d = (p_1^d, p_2^d, ..., p_{|S|}^d)$, respectively. We first show that $p_i^d \leq p_i$ for all $i \in [1,|S|]$. Denote the SINR threshold levels corresponding to the rates $r_i$ and $r_i^d$ by $\gamma_i$ and $\gamma_i^d$, respectively. Since higher rate corresponds to a higher SINR threshold value and $r_i^d \leq r_i$, $\gamma_i^d \leq \gamma_i$. Therefore, any power vector satisfying the SINR requirements for the SINR threshold levels $\gamma_i^d, i \in [1, |S|]$ also satisfies these requirements for $\gamma_i, i \in [1, |S|]$. Suppose that $p_k^d > p^k$ for an arbitrary link $k$, then $\textbf{p}^d$ cannot be component-wise minimum power vector. This is a contradiction. 

If $\textbf{r}^d$ is infeasible due to the maximum power constraint given in Eq. (\ref{eq:pmax}), then there exists at least one link $k$ for which $p_k^d > p_{max}$. Since $p_k \geq p_k^d$, $p_k > p_{max}$. Hence, $\textbf{r}$ does not satisfy maximum power constraint either.

If $\textbf{r}^d$ is infeasible due to the energy constraint given in Eq. (\ref{eq:ei}), then there exists at least one link $k$ for which $t_k^d p_k^d> e_k$. Let $\alpha$ denote the ratio of $r_k$ and $r_k^d$. Since SINR-rate mapping function is an increasing function with a non-negative second derivative as described in Section \ref{system_model}, the transmission power $p_k$ will be greater than $\alpha p_k^d$. Therefore, $t_k p_k  > \frac{R_k}{ \alpha r_k^d} \alpha p_k^d  = t_k^d p_k^d > e_k$ so $r$ is infeasible due to the energy constraint.

\end{proof}

\label{lemma1}
%Lemma 1
\begin{lemma}
Let $t^d$ and $t$ are the transmission times corresponding to the rate index vectors $\textbf{r}^d$ and $\textbf{r}$, respectively. Then, $t^d \geq t$ is always true.
\end{lemma}

\begin{proof}
We know that, $t = max_{i\in[1,|S|]}t_i = max_{i\in[1,|S|]}\frac{R_i}{r_i}$. Since we are considering the case where $r_i^d \leq r_i$ for all  ${i\in[1,|S|]}$, we can simply show that $t = max_{i\in[1,|S|]}\frac{R_i}{r_i} \leq max_{i\in[1,|S|]}\frac{R_i}{r_i^d} = t^d$. Hence, $t^d \geq t$.
\end{proof}
%
%Lemma 2
\label{lemma2}
\begin{lemma}
Suppose that node $j$ has the maximum time slot length among concurrently transmitting $|S|$ nodes, i.e. $j = argmax_{i\in[1,|S|]}\frac{R_i}{r_i}$. Then, increasing the transmission rate of any node $k\neq j$, where ${k, j \in[1,|S|]}$, does not decrease the time slot length $t$.
\end{lemma}
\begin{proof}
$t=max_{i\in[1,|S|]}\frac{R_i}{r_i}=\frac{R_j}{r_j} \geq \frac{R_k}{r_k}$  $\forall{k\neq j}$, where ${k, j \in[1,|S|]}$. Suppose that, for an arbitrary link $k \neq j$, we increase $r_k$ to ${r_k}^\prime$. Then $\frac{R_k}{{r_k}^{\prime}} \leq \frac{R_k}{r_k} \leq \frac{R_j}{r_j} = t$. Hence, increasing the transmission rate of any node $k \neq j$ does not decrease $t$.
\end{proof}

\begin{theorem}
Algorithm stated in Section IV gives the optimal result.
\end{theorem}

\label{theorem proof}
\begin{proof}
Suppose that LTTF returns the $\textbf{r}^{opt}$ as an optimal solution and the $j$-th link has the maximum transmit time among all the sensor nodes. Now, consider that $\textbf{r}^{opt}$ is not an optimal solution. This means that there exists another optimal rate index vector $\textbf{r}^{**}$ such that $t^{**} \leq t^{opt}$. 

Due to Lemma 2, the transmission times corresponding to all the descendants of $\textbf{r}^{opt}$ are greater or equal to that corresponding to $\textbf{r}^{opt}$ itself, i.e. $t^d \geq t^{opt}$.  This means that $\textbf{r}^{**}$ cannot be a descendant of $\textbf{r}^{opt}$. Also due to Lemma 3, increasing the rate value of any other node than $j$ cannot decrease $t^{opt}$. 

We now introduce the new rate index vector $\textbf{r}^{opt,+}$ in which only $r_j$ is increased by one level, i.e. $\textbf{r}^{opt,+} = (r_1^{opt}, r_2^{opt}, ..., r_j^{opt,+}, ... , r_{|S|}^{opt})$. Since LTTF returns $\textbf{r}^{opt}$, we know that $\textbf{r}^{opt,+}$ is infeasible. Due to Lemma  1, for at least one sensor $k \neq j$, $r_k^{**} < r_k^{opt}$ must hold for $\textbf{r}^{**}$ to be feasible. We also know that at each iteration of the LTTF algorithm, the rate value of the sensor node with the maximum transmission time is increased by one level. Therefore at some iteration of the algorithm, the $k$-th node is evaluated as $k=argmax\frac{R_i}{r_i}$  $\forall{i} \in [1, |S|]$ and $r_k^{**}$ increased by one level. It follows that, $t^{**} > t_j^{opt} = t^{opt}$, which is a contradiction.

\end{proof}

\section{Simulation Results} \label{simulationresults}

%In this section, we evaluate the performance of the proposed optimal and discrete rate power control algorithm compared to the previously proposed continuous rate algorithms. We first give the simulation parameters, describe the simulation setup and the characteristics of the wireless channel, then discuss the final results and demonstrate the superiority of the proposed power control algorithm.

Simulations are performed in MATLAB. The sensor nodes and the controller are uniformly distributed within a square area. The figures show the average of 100 independent random network topologies. The packet generation period and the packet transmission requirement of each node are randomly chosen from the sets $[1, 2, 4, 8]$ and $[50, 100]$, respectively. Number of controllers is fixed and equal to $3$. %The y axis values in all figures are obtained by averaging the results over 100 independent runs. 
The continuous rate model adopts Shannon's channel capacity formulation for an AWGN wireless channel in the calculation of the maximum achievable rate as a function of SINR and is denoted by \emph{cont} in the legends of the figures. Discrete rate model adopts $4$ discrete rate levels corresponding to 4 SINR levels equal to $[-\infty, 10, 20, 30]$ dB and $8$ discrete rate levels corresponding to $8$ SINR levels equal to $[-\infty, 0, 5, 10, 15, 20, 25, 30]$ dB  calculated by using Shannon's capacity formulation, and denoted by \emph{disc4} and \emph{disc8}, respectively, in the legends of the figures. The channel attenuation is calculated by using Rayleigh fading with scale parameter set to the mean power level determined by using the large scale statistics modeled as
\begin{equation} \label{eq:pathloss}
PL(d)=PL(d_0)+10 \alpha \log (d/d_0)+ Z,
\end{equation}
where $d$ is the distance between the node and the controller, $PL(d)$ is the path loss at distance $d$, $PL(d_0)$ is the path loss at reference distance $d_0$, $Z$ is the Gaussian random variable with zero mean and standard deviation $\sigma_z$. All the simulation parameters are given in Table \ref{table:parameters}.

\begin{table}[h] 
\centering
\caption{\small Simulation Parameters}
\label{table:parameters}
\begin{tabular}{||c|c||c|c||} 
  \hline                       
  $PL(d_0)$  & $70$~dB & $p_{max}$ & $250$~mW  \\ \hline
  $\alpha$  &  $3.5$ & $\sigma_z$ & $4$~dB\\ \hline
    $N_0$ & $10^{-8}$~W/Hz &  $W$ & $100$~MHz\\ \hline
 \end{tabular}
\end{table}

Fig. \ref{fig:numOfNodes} shows the normalized maximum total active length of SNA-MLA and SNA-MUA scheduling algorithms for continuous and discrete rate transmission models for different number of nodes at $5$ nodes$/m^2$ density. The normalized maximum total active length is defined as the ratio of the maximum total active length of the scheduling algorithm to that of the optimal solution for the continuous rate transmission model. The reason for normalizing the maximum total active length of even the discrete rate by that of the continuous rate transmission is to observe the performance of the discrete rate model compared to the continuous rate model. We observe that the discrete rate model with $8$ levels performs very close to the continuous rate model. Moreover, as the number of the rate levels increases, the performance of the discrete rate model gets closer to that of the continuous rate model. Furthermore, as the number of nodes increases, the normalized schedule length increases, mainly due to the increasing number of the combinations of the links for concurrent transmissions. SNA-MLA performs better than SNA-MUA since it checks an exponential number of link combinations.

\begin{figure}[h] 
  \begin{center} 
    \includegraphics*[scale=0.5,viewport=370 120 920 550]{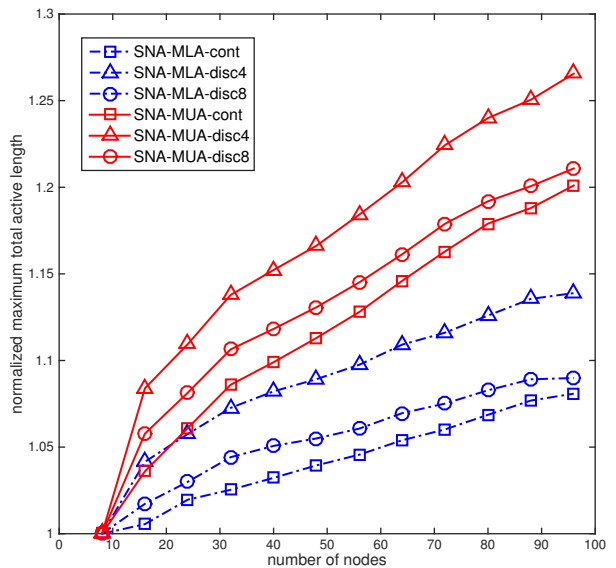}
  \end{center}  
  \caption{\small Normalized maximum total active length of SNA-MLA and SNA-MUA scheduling algorithms for continuous and discrete rate transmission models for different number of nodes at $5$ nodes$/m^2$ density.}  
  \label{fig:numOfNodes}
\end{figure}

Fig. \ref{fig:networkDensity} shows the normalized maximum total active length of SNA-MLA and SNA-MUA scheduling algorithms for continuous and discrete rate transmission models at different network densities in a network of 100 nodes. Similarly, as the number of the rate levels increases, the performance of the discrete rate model gets closer to that of the continuous rate model. The
difference between the performance of these algorithms decreases at low and high node densities. This is
mainly due to the decreasing number of the combinations of the links for concurrent transmissions: At
low node densities, the number of link combinations suitable for concurrent transmission is limited since
most of the nodes are suitable for concurrent transmission, while causing only slight interference, being
separated enough from each other. At high node densities, the number of node subsets that are separated
enough from each other for concurrent transmission is limited due to the massive interference caused
among each other.

\begin{figure}[h] 
\begin{center} 
    \includegraphics*[scale=0.55,viewport=380 150 970 550]{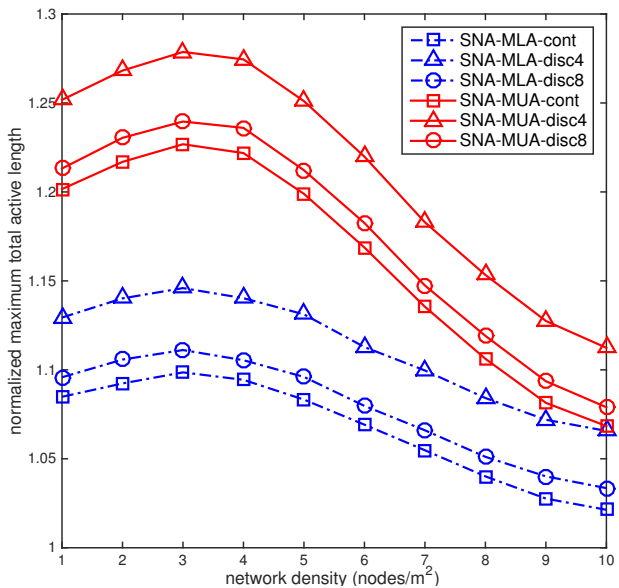}
  \end{center}  
  \caption{\small Normalized maximum total active length of SNA-MLA and SNA-MUA scheduling algorithms for continuous and discrete rate transmission models at different network densities in a network of $100$ nodes.}  
  \label{fig:networkDensity}
\end{figure}

\section{Conclusion} \label{sec:conclusion}

In this paper, we propose a polynomial time power control and rate adaption algorithm for ultra-reliable M2M control applications employing discrete
rate transmission model. The algorithm is based on the initialization with the vector of the lowest possible rates satisfying the delay
constraints of the nodes and then checking the feasibility of the lowest possible number of rate allocation
vectors for satisfying energy, maximum power and SINR constraints. We proved the optimality of the proposed algorithm. The extensive simulations of the algorithm in combination with the previously proposed scheduling algorithms demonstrate that the performance of the discrete rate transmission model is very close to that of the continuous rate model with high enough discrete rate levels, and robust to the changes in the network size and density.  
In the future, we are planning to extend this framework for multi-hop M2M applications and next generation
cellular networks.

\bibliographystyle{IEEEtran}
\bibliography{conf_ref}

% that's all folks
\end{document}